\documentclass[11pt, a4paper]{article}

\usepackage{fullpage}
\usepackage{amsmath,amssymb,srcltx}
\usepackage{amsmath,amssymb,amsthm} 
\usepackage{graphicx}
\usepackage{amscd}
\usepackage{amsmath}
\usepackage{amsfonts}
\usepackage{amssymb}
\usepackage{setspace}
\usepackage{enumerate}         
\usepackage{color}
\usepackage{url}
\usepackage{amsthm}
\usepackage{hyperref}
\usepackage{amsmath}
\usepackage{bm}
\usepackage{xy}
\usepackage{stmaryrd}
\usepackage{fancyhdr}
\usepackage{authblk}

\theoremstyle{plain}
\newtheorem{theorem}{Theorem}[section]

\newtheorem{corollary}[theorem]{Corollary}

\newtheorem{lemma}[theorem]{Lemma}

\newtheorem{example}[theorem]{Example}

\newtheorem{proposition}[theorem]{Proposition}

\newtheorem{definition}[theorem]{Definition}

\newtheorem{assumption}[theorem]{Assumption}
\theoremstyle{remark}
\newtheorem{remark}[theorem]{Remark}

\numberwithin{equation}{section}

\newcommand{\RR}{\mathbb{R}}

\newcommand{\NN}{\mathbb{N}}

\newcommand{\lzs}{{{L}}^0}

\newcommand{\lo}{{{L}}^1(\massP)}

\newcommand{\fps}{(\Omega,\cF,(\cF_t)_{\tT},\massP)}

\newcommand{\massP}{\mathbf{P}}

\newcommand{\massE}{\mathbb{E}}
\newcommand{\p}{\massP}

\newcommand{\la}{\lambda}
\newcommand{\cA}{\mathcal{A}}
\newcommand{\cB}{\mathcal{B}}
\newcommand{\cC}{\mathcal{C}}
\newcommand{\cD}{\mathcal{D}}

\newcommand{\cF}{\mathcal{F}}

\newcommand{\cZ}{\mathcal{Z}}

\newcommand{\wtcD}{\widetilde {\mathcal{D}}}

\renewcommand{\i}{\infty}

\newcommand{\tT}{0 \leq t \leq T}

\newcommand{\wt}[1]{{\widetilde{#1}}}
\newcommand{\wh}[1]{{\widehat{#1}}}
\newcommand{\lims}[2]{ \lim_{#1 \to \i} #2_{#1}}

\newcommand{\seq}[2]{\{{#1}^{#2}\}_{#2=1}^\infty}
\newcommand{\sequ}[2]{\{{#1}_{#2}\}_{#2=1}^\infty}
\newcommand{\seqnet}[3]{\{{#1}_{#2}\}_{#2 \in #3}}

\newcommand{\Econd}[2]{\massE\left[\left.#1\right|#2\right]}        
\newcommand{\Ex}[2]{\massE^{#1}\left[#2\right]}                     
\newcommand{\Excond}[3]{\massE^{#1}\left[\left.#2\right|#3\right]}  
\newcommand{\Es}[1]{\massE \left[ #1 \right]}

\providecommand{\keywords}[1]{\textbf{{Keywords.}} #1}
\providecommand{\subjclass}[1]{\textbf{{MSC2010.}} #1}

\begin{document}

\title{Utility maximization problem under transaction costs: optimal dual processes and stability\thanks{The authors gratefully acknowledge financial support from the Austrian Science Fund (FWF) under grant P25815 and from the European Research Council under ERC Advanced Grant 321111 and from the University of Vienna under short-term grand abroad (KWA). This work is partially done during the visit of L. Gu and J. Yang at CMAP, \'Ecole Polytechnique, hosted by Prof. N. Touzi, who are grateful acknowledged.} }
\author[1,3]{Lingqi Gu}
\author[2]{Yiqing Lin\thanks{corresponding author: yiqing.lin@polytechnique.edu}}
\author[2]{Junjian Yang}
\affil[1]{\small Fakult\"at f\"ur Mathematik, Universit\"at Wien, Oskar-Morgenstern Platz 1, A-1090 Wien, Austria}
\affil[2]{\small Centre de Math\'ematiques Appliq\'ees, \'Ecole Polytechnique, F-91128 Palaiseau Cedex, France}
\affil[3]{\small College of Mathematics and informatics, Fujian Normal University, 350117 Fuzhou, China}
\date{\small \today}

\maketitle

\begin{abstract}\noindent
\indent This paper discusses the num\'eraire-based  utility maximization problem in markets with proportional transaction costs. 
In particular, the investor is required to liquidate all her position in stock at the terminal time. 
We first observe the stability of the primal and dual value functions as well as the convergence of the primal and dual optimizers when perturbations occur on the utility function and on the physical probability. We then study the properties of the optimal dual process (ODP), that is, a process from the dual domain that induces the optimality of the dual problem. When the market is driven by a continuous process, we construct the ODP for the problem in the limiting market by a sequence of ODPs corresponding to the problems with small misspecificated parameters. Moreover, we prove that this limiting ODP defines a shadow price.

\end{abstract}

\noindent \keywords{Utility maximization problem, Transaction costs, Stability, Optimal dual processes, Shadow price processes}\\

\noindent \subjclass{91G80, 93E15, 60G48}

\maketitle

\section{Introduction}

The utility maximization problem with constant proportional transaction costs is thoroughly studied recently. In this paper, we restrict ourselves to consider such a problem with a num\'eraire-based general model, in which the investor aims to solve the maximization problem on the expected utility over her terminal wealth. In this market, the stock price is driven by a process $S$ which is not necessarily a semimartingale. However, the investor buys the stock at the price $S$ but receives only $(1-\lambda)S$ when selling them, where $\lambda$ denotes the constant proportional transaction costs. The investor trades the stock $S$ with admissible strategies and is required to liquidate all her position in stock at the terminal time $T$. The utility maximization problem on the positive half-line is formulated with a set $\mathcal{A}(x)$ of admissible trading strategies:
 \begin{equation}\label{eq019}
   \massE\big[U\big(V_T^{liq}(\varphi)\big)\big] \to \max!, \qquad \varphi=(\varphi^0, \varphi^1)\in\mathcal{A}(x), \\
 \end{equation}
 where $\varphi^0$ and $\varphi^1$ denote her position in bond and in stock, respectively. \\
 
 For frictionless market models beyond semimartingales, in general, arbitrage opportunities exist, so that the utility maximization problem could not be wellposed. 
  However, the presence of transaction costs could exclude such strategies: for example, Guasoni \cite{Gua06} has showed that the fractional Black-Scholes model is arbitrage-free when arbitrarily small proportional transaction costs were taken into consideration. In this case, optimal strategies can be found for the maximization problem with a finite indirect utility function \cite{Gua02}. With proportional transaction costs $\lambda$, the duality theory for utility maximization dates back to the seminal work of \cite{CK96, CW01} when $S$ is driven by an It\^o process, which has been afterward extended by \cite{CS16duality, CSY17} to the general framework when the utility function supports only the positive half-line. Besides, we refer the reader to \cite{DPT01, Bou02, BM03, CO11, Yu17} in the similar context but with possibly multivariate utility functions.\\
  
In \cite{CS16duality}, the solution couple $\wh\varphi=(\wh\varphi^0, \wh\varphi^1)$ for the primal problem is associated with the solution of the following dual problem
  \begin{equation}\label{eq029}
   \massE[V(yY^0_T)] \to \min!, \qquad Y=(Y^0, Y^1)\in\mathcal{B}(1). \\
  \end{equation}
Indeed, the problem (\ref{eq019}) is solved under the assumption of the existence of {\it consistent price systems} introduced in \cite{JK95}. Precisely, a consistent price system is a couple of processes $Z=(Z^0, Z^1)$ consisting of a strictly positive martingale $Z^0$ and a positive local martingale $Z^1$ such that 
  $$
  S^Z_t:=\frac{Z^1_t}{Z^0_t}\in [(1-\lambda)S, S],\qquad 0\leq t\leq T, \quad {\rm a.s.}
  $$ 
  These processes take the role of ``{\it equivalent local martingale measures}'' (or  ``{\it equivalent local martingale deflators}'') when the same problem is considered under the assumptions of (NFLVR) (or (NUPBR)) in a frictionless  market. Instead of the Fatou limit argument in \cite{KS99}, the domain $\mathcal{B}(1)$ of the dual problem (\ref{eq029}) is a collection of all optional strong supermartingales which are limits of consistent prices systems in the sense \cite{CS16strong} 
\begin{equation}\label{eq039}
 Z^n_\tau\longrightarrow Y_\tau, \qquad {\rm in}\ \massP,\qquad {\rm for\ all}\ 0\leq \tau\leq  T.
 \end{equation}
 We point out that the duality result in \cite{CS16duality, CSY17} is not only static but also dynamic, which means that one can find an {\it optimal dual process} (for short, ODP) $\widehat{Y}=\big(\widehat Y^0, \widehat Y^1\big)$ in $\mathcal{B}(1)$ whose first component $Y^0$ attains the optimality of (\ref{eq029}). In addition, if the process 
 $\widehat{Y}^0$ is a local martingale, then a {\it shadow market} can be defined as in \cite{CS16duality} via 
 $$\widehat{S}:=\frac{\widehat{Y}^1}{\widehat{Y}^0}.$$  
 
 Conceptually, the so-called shadow market is a frictionless market driven by some price process $\widetilde{S}$ lying between the bid-ask spread $[(1-\lambda)S, S]$. If the investor trades with $\widetilde{S}$ frictionlessly instead of with $S$ under transaction costs, the utility maximization problem (\ref{eq019}) is solved by the same trading strategy which yields the same optimality. As far back as the work \cite{CK96}, the relation between the ODP and the {\it shadow price process} has been investigated. It is worth mentioning that such process always exists if the probability space is finite, see \cite{KMK11}. However, it could fail to exist in a general context (see counterexamples in \cite{BCKMK13, CMKS14}). Concerning the utility maximization problem (\ref{eq019}), the sufficient conditions for the existence of a local martingale type ODP are studied in \cite{CSY17, CPSY16}, so that in their cases, a shadow price can be constructed by the result of \cite{CS16duality}.\\

  In the present paper, we concern ourselves with the sensitivity of the duality result obtained in \cite{CS16duality, CSY17}. In particular, we consider the question: \\
  
  {\it\indent How the small changes on input parameters (initial wealth, the inverstor's preferences, etc...) exert influence on the optimal strategy and the optimality of the problem (\ref{eq019})?\\} 

In the framework without transaction costs, the research on the stability of a problem similar to (\ref{eq019}) dates back to \cite{JN04} (see also \cite{CR07}), in which Jouini and Napp have proved the $L^p$-convergence ($p \geq 1$) of the optimal investment-consumption strategy when the utility function was perturbed in the It\^o-process market. Afterwards, this result has been generalized by Larsen \cite{Lar09} to consider a similar problem in the general continuous semimartingale market. Similar to \cite{JN04}, Larsen has imposed a growth condition on the sequence of utility functions. On the other hand, if only the pathwise convergence of the utility functions was assumed, he has a counterexample for its non-sufficiency to deduce the continuity of the value function.
 (See also \cite{KZ11} for detailed discussion on this uniform integrability (UI) type condition). In \cite{KZ11}, Kardaras and \v{Z}itkovi\'{c} allowed variations not only on the preferences of the investor but also on the subjective probabilities and have obtained the $L^0$-stability for the primal and dual optimizers and the continuity of the value functions and their derivatives. More recently, Xing has investigated this problem in \cite{Xin17} when the target exponential utility function is approximating by functions defined either on $\mathbb{R}$ or on the half-line. 
Another type of stability problem is studied in \cite{LZ07, Fre13, BK13, Wes16} (see also \cite{MW13}), where misspecifications of the model are denoted by variations of the coefficients for the risky process. \\

The present paper is organized as follows.
In section 2 we introduce the formulation of the model and the duality result for solving the problem (\ref{eq019}) with a c\`adl\`ag price process $S$. 
Then, we conduct a sensitivity analysis in Section 3, in which we adapt the (UI) condition in \cite{KZ11} to the context under transaction costs. For a sequence of perturbed initial wealth, of the investor's preferences and of the market's subjective probabilities, we find the $L^0$-stability of the primal and dual optimizers and the continuity of the primal and dual value functions along with their derivatives. Section 4.1 devotes to the study on the dual domain of the problem (\ref{eq029}) and the property of ODPs, which complements the results in \cite{CS16duality, CSY17}. In particular, we show that for $Y\in \mathcal{B}(1)$, the first component $Y^1$ is a local martingale if and only if $Y^0$ is a local martingale. Moreover, we provide a simple explanation of why the ODP must be a local martingale in \cite{CSY17} when $S$ is continuous and the running liquidation value of the optimal wealth process is bounded away from 0 (see (\ref{spodp})). Another important observation is that the optimal dual process is not unique, which is in contrast to the result in the classical frictionless theory. Thanks to W. Schachermayer, a counterexample is constructed with a simple time-changed geometric Brownian motion. 
In addition to the static stability result in Section 3, we investigate moreover the so-called dynamic stability in Section 4.2 based on our studies on ODPs in the previous subsection. Mathematically, if $\{\wh Y^n\}_{n \in \NN}$ of ODPs associated with the perturbed problems, then there exists at least a convex combination of $\{\wh Y^n\}_{n \in \NN}$ that admits a limit $\widehat{Y}(x; U,\p)$ in the sense of (\ref{eq039}). This limit must be an ODP of the limiting problem. When the situation of \cite{CSY17, CPSY16} is under consideration, the dynamic stability result suggests a possible way for constructing a shadow price process for the limiting problem from the ones of the perturbed problems, which is studied in Section 4.3.

\section{Formulation of the utility maximization problem}
In this section, we introduce a num\'eraire-based market model with transaction costs and existing duality results obtained in \cite{CS16duality}. \\

Fix a finite time horizon $T > 0$. The market involves proportional transaction costs  $0<\lambda< 1$,
in other words, a financial agent has to buy stock shares at the higher ask price $S_t$, and only receives a lower bid price $(1-\la)S_t$ when selling them.
\begin{assumption}\label{Slb}
 The stock price process  $S=(S_t)_{0 \leq t \leq T}$ is strictly positive, c\`adl\`ag, adapted and based on a filtered probability space $\fps$ satisfying the usual conditions of right continuity and saturatedness.  Additionally, we assume $\mathcal{F}_{T-}=\mathcal{F}_T$ and
 $S_{T-} =S_T $.
\end{assumption}
\begin{remark}
 The additional assumptions $\mathcal{F}_{T-}=\mathcal{F}_T$ and
 $S_{T-} =S_T $ are to avoid special notation for possible trading at the terminal time $T$. That is, without loss of generality, we assume that  the price $S$ does not jump at the terminal time $T$, while the investor can still liquidate her position in the stock shares, so that we may let $\varphi_T^1=0.$ For more details on these assumptions we refer to,  e.g., \cite[Remark 4.2]{CS06} or \cite[p.~1895]{CS16duality}.
\end{remark}

\begin{definition}
 A trading strategy, modeling the holdings in units of the bond and of the stock, is an $\mathbb{R}^2$-valued, predictable, finite variation process $\varphi = (\varphi^0_t,\varphi^1_t)_{0\leq t\leq T}$
   such that the following {\it self-financing} constraint is satisfied:
  \begin{equation*} 
   \int_s^td\varphi_u^0 \leq -\int_s^tS_ud\varphi_u^{1,\uparrow}+ \int_s^t(1-\lambda)S_ud\varphi_u^{1,\downarrow},
  \end{equation*}
  for all $0\leq s<t\leq T$, where $\varphi^1=\varphi^{1,\uparrow}-\varphi^{1,\downarrow}$ denotes the canonical decompositions of $\varphi^1$ into the difference of two increasing processes. 
\end{definition}

\begin{remark}
 Any finite variation process $\varphi$ is l\`adl\`ag. 
 As pointed out in \cite{CSY17}, we may assume trading strategies to be c\`adl\`ag, if the price process is continuous.
\end{remark}

 \begin{definition}  \label{strategy}
 Fix the level $0<\lambda<1$ of transaction costs. For trading strategy $\varphi = (\varphi^0_t,\varphi^1_t)_{0\leq t\leq T}$, we define its {\normalfont liquidation value} $V_t^{liq}(\varphi)$ at time $t\in [0, T]$ by
 \begin{equation*} 
   V_t^{liq}(\varphi) := \varphi^0_t+(\varphi_t^1)^+(1-\la)S_t-(\varphi_t^1)^-S_t.
 \end{equation*}

 A trading strategy $\varphi$ is called {\normalfont admissible,} if $V_t^{liq}(\varphi)\geq 0$ for all $0\leq t\leq T$. \\
 
 For $x>0$, we denote by $\mathcal{A}(x)$ the set of all admissible, self-financing trading strategies $\varphi = (\varphi^0_t,\varphi^1_t)_{0\leq t\leq T}$, 
   starting with the initial endowment $(\varphi_{0}^0,\varphi_{0}^1)=(x,0)$.
 Denote by $\mathcal{C}(x)$ the convex subset of terminal liquidation values
  \begin{equation*} 
   \mathcal{C}(x):= \left\{V^{liq}_T(\varphi):\,\,\varphi\in\mathcal{A}(x)\right\}\subseteq L_+^0(\massP),
  \end{equation*}
  which equals the set $\left\{\varphi_T^0:\,\,\varphi=(\varphi^0,\varphi^1)\in\mathcal{A}(x),\ \varphi_T^1=0 \right\}$ as defined in \cite{CS16duality}.
\end{definition}

We use a utility function to model the agent's preferences.
\begin{assumption}\label{utifunc}
Let \ $U:\RR_+\to\RR$ be a strictly increasing, strictly concave and smooth function, satisfying the 
Inada conditions $U'(0)=\infty$ and $U'(\infty)=0$, 
as well as the condition of ``reasonable asymptotic elasticity" introduced in \cite{KS99}, i.e., 
\begin{equation*}
 \textnormal{AE}(U):= \limsup\limits_{x\to\i} \frac{xU'(x)}{U(x)} < 1.
\end{equation*}
\end{assumption}

Given the initial endowment $x>0$, the agent wants to maximize his {\it expected utility} at the terminal time $T$:
 \begin{equation}\label{J5}
   \massE\big[U\big(V_T^{liq}(\varphi)\big)\big] \to \max!, \qquad \varphi\in\mathcal{A}(x). \\
 \end{equation}

In order to avoid the trivial case, we have the following assumption. 
\begin{assumption}\label{uf}
 Suppose that $$ \sup_{ g\in\mathcal{C}(x)}\massE[U(g)] <\i,$$ for some $x>0.$ 
\end{assumption}
 
The duality theorem on the utility maximization in \cite{CS16duality} requires an assumption on the existence of consistent prices systems, which is  similar to the existence of equivalent martingale measures or some similar weaker conditions for frictionless markets.

\begin{definition} \label{CPS}
Fix $0<\la <1$.  
A \textnormal{$\la$-consistent price system} is a two dimensional strictly positive process $Z=(Z^0_t,Z^1_t)_{0\le t\le T}$ 
 with $Z^0_0=1$, that consists of a martingale $Z^0$ and a local martingale $Z^1$ under $\massP$ such that 
  \begin{equation*} 
    \widetilde{S}_t:=\frac{Z^1_t}{Z^0_t} \in [(1-\la)S_t, S_t],\qquad 0\leq t\leq T,\qquad {\rm a.s.} 
  \end{equation*}
We denote by $\mathcal{Z}^{\lambda} (S)$ the set of $\la$-consistent price systems and we say that $S$ satisfies the condition $(CPS^\la)$ of {admitting a $\la$-consistent price system}, if $\mathcal{Z}^{\lambda}(S)$ is nonempty. 
\end{definition}

\begin{definition}
We say that $S$ satisfies the condition $(CPS^\la)$ \textnormal{locally}, 
if there exists a strictly positive process $Z$ and a sequence $\{\tau_n\}_{n\in\NN}$ of $[0,T] \cup\{\i\}$-valued stopping times, increasing to infinity, 
such that each stopped process $Z^{\tau_n}$ defines a $\lambda$-consistent price system for the stopped process $S^{\tau_n}$. We call this processes $Z$ local $\lambda$-consistent price system and denote by $\mathcal{Z}^{loc,\lambda}$ the set of all such processes.  
\end{definition}

\begin{assumption}\label{localCPS}
 The stock price process $S$ satisfies $(CPS^{\mu})$ locally,  for all $0 < \mu < \lambda$.
\end{assumption}

\begin{definition}[Optional strong supermartingale]
A real-valued stochastic process $X = (X_t)_{\tT}$ is called optional strong
supermartingale, if
\begin{enumerate}[(1)]
 \item  $X$ is optional;
 \item  $X_{\tau}$ is integrable for every $[0, T]$-valued stopping time $\tau$;
 \item  For all stopping times $\sigma$ and $\tau$ with  $ 0 \leq \sigma  < \tau \leq T$, we have
         $$X_{\sigma} \geq \Econd{X_{\tau}}{\cF_{\sigma}}.$$
\end{enumerate}
\end{definition}

\begin{remark}
The notion of such processes introduced by Mertens in \cite{Mer72} is a generalization of  c\`adl\`ag supermartingales.  The readers are also referred to \cite[Appendix I]{DM82} for more properties of these processes. 
\end{remark}

\begin{definition} \label{defB}
We denote by $\cB(y)$ the set of all optional strong supermartingale deflators, which are pairs of nonnegative {{optional strong supermartingales}} $Y=(Y^0_t,Y^1_t)_{0\leq t\leq T}$
 such that $Y^0_0=y$, $Y^1=Y^0\widetilde S$ for some $[(1-\lambda)S,S]$-valued process $\widetilde S = (\widetilde S_t)_{0\leq t\leq T}$, 
 and $Y^0(\varphi^0+ \varphi^1 \wt S)= Y^0\varphi^0+Y^1\varphi^1$ is a nonnegative
optional strong supermartingale for all $(\varphi^0,\varphi^1)\in\mathcal{A}(1)$. 
Accordingly, define
\begin{align*} 
\cD(y)=\Big\{ Y^0_T:\, (Y^0,Y^1)  \in  \cB(y)\Big\},  \quad \mbox{for} \ y >0.
\end{align*}
\end{definition}

\begin{remark}
In the frictionless case, we do not need to pass to the optional version of supermartingale deflator, since the wealth process associated with a trading strategy is always c\`adl\`ag as a stochastic integral, which is not necessarily true in the new context. 
\end{remark}

The polarity between the sets $\mathcal{C}:=\cC(1)$ and $\mathcal{D}:=\mathcal{D}(1)\subseteq L^0_+(\massP)$ are established in \cite[Lemma A.1]{CS16duality}, 
  which means that $\mathcal{C}$ and $\mathcal{D}$ in this new context satisfy verbatim the conditions listed in \cite[Proposition 3.1]{KS99}. 
As a result, the following duality result in \cite{CS16duality} is straightforward by following the lines of \cite{KS99}.

\begin{theorem}[Duality Theorem, {\cite[Theorem 3.2]{CS16duality}}]\label{thhurt} 
Under Assumptions \ref{Slb}, \ref{utifunc}, \ref{uf}, \ref{localCPS},
define the primal and dual value functions as
\begin{align}
  u(x):=\sup\limits_{g\in\mathcal{C}(x)} \massE[U(g)], \notag \\
  v(y):=\inf\limits_{h\in\mathcal{D}(y)} \massE[V(h)], \label{geht}
\end{align}
where  $$V(y):= \sup_{x>0}\{U(x)-xy\}, \quad y>0, $$ 
 is the conjugate function of $U$. 
 
Then, the following statements hold true.
\begin{enumerate}[(i)]
 \item The functions $u(x)$ and $v(y)$ are finitely valued, for all $x$, $y>0$, and mutually conjugate
        \begin{equation*}
             v(y)=\sup\limits_{x>0} [u(x)-xy], \quad u(x) =\inf\limits_{y>0} [v(y)+xy].
        \end{equation*}
       The functions $u$ and $v$ are continuously differentiable and strictly concave (respectively, convex) and satisfy
        \begin{equation*}
             u'(0) =-v'(0)=\i, \qquad u'(\i)=v'(\i)=0.
        \end{equation*}
        
 \item For all $x$, $y>0$, the solutions $\widehat{g}(x)\in\mathcal{C}(x)$  and $\widehat{h}(y)\in\mathcal{D}(y)$  exist, 
        are unique and take their values a.s.~in $(0,\i)$.
       There are $\big(\widehat{\varphi}^0(x),\widehat{\varphi}^1(x)\big)\in\mathcal{A}(x)$ and $\big(\widehat{Y}^0(y),\widehat{Y}^1(y)\big)\in\mathcal{B}(y)$
        such that 
   \begin{align}\label{optimizer}
      V_T^{liq}\big(\widehat{\varphi}(x)\big)=\widehat\varphi^0_T=\widehat{g}(x)\quad\mbox{and}\quad \widehat{Y}^0_T(y)=\widehat{h}(y). 
   \end{align}

 \item  If $x>0$ and $y>0$ are related by $u'(x)=y$, or equivalently by $x=-v'(y)$, 
        then $\widehat{g}(x)$ and $\widehat{h}(y)$ are related by the first order conditions
        \begin{equation} \label{Y=U'(X)}
            \widehat{h}(y)=U'\big(\widehat{g}(x)\big) \quad\mbox{and}\quad \widehat{g}(x) = -V'\big(\widehat{h}(y)\big),
        \end{equation}  
        and we have 
        \begin{equation*} \label{xy=EXY}
           \massE\big[\widehat{g}(x)\widehat{h}(y)\big]=xy.
        \end{equation*}           
        In particular, $\widehat{\varphi}^0(x)\widehat{Y}^0(y)+\widehat{\varphi}^1(x)\widehat{Y}^1(y)$ is a $\massP$-martingale for all  $\big(\widehat{\varphi}^0(x),\widehat{\varphi}^1(x)\big)\in\mathcal{A}(x)$ and $\big(\widehat{Y}^0(y),\widehat{Y}^1(y)\big)\in\mathcal{B}(y)$ satisfying \eqref{optimizer}.
\item Finally, we have $$v(y)= \inf_{(Z^0,Z^1)\in \cZ^{loc,\lambda}}\Es{V(yZ^0_T)}.$$
\end{enumerate}
\end{theorem} 

\begin{remark}
Similar to \cite[Proposition 3.2]{KS99} (see also \cite[Appendix A]{CS16duality}), the infimum of $v(y)$ defined in (\ref{geht}) could be approximated by the element from $\{Z^0_T: Z\in \mathcal Z^{loc, \lambda}\}$, i.e., 
\begin{equation}\label{vequiv}
v(y):=\inf\limits_{h\in\mathcal{D}(y)} \massE[V(h)]=\inf\limits_{Z\in\mathcal{Z}^{loc, \lambda}} \massE[V(yZ^0_T)].
\end{equation}
\end{remark}

We clarify here notations for our argument of stability in the following sections. For the utility maximization problem in the market under the physical probability $\massP$, we employ
$$\cB(y)= \cB(y, \p),\ \cD(y)= \cD(y, \p),\  \cZ^\lambda= \cZ^\lambda(\p),\  \cZ^{loc,\lambda}= \cZ^{loc,\lambda}(\p).$$

Moreover, we call the optional strong supermartingale deflator $\big(\widehat{Y}^0(y),\widehat{Y}^1(y)\big)\in\mathcal{B}(y)$ inducing $\widehat{Y}^0_T(y)=\widehat{h}(y)$ {\it optimal dual process (ODP)}.  
Sometimes, we use the calibrated processes $\wh Y:=(\wh Y^0, \wh Y^1)\in \mathcal B(1)$, and by abuse of definition, we still call it ODP. 
Similarly, we call the trading strategy $\big(\widehat{\varphi}^0(x),\widehat{\varphi}^1(x)\big)\in\mathcal{A}(x)$ satisfying $V_T^{liq}\big(\widehat{\varphi}(x)\big)=\widehat{g}(x) $ {\it optimal primal process (OPP)}.\\

When solving a utility maximization problem in the frictional market, we often wonder whether this market can be replaced by a frictionless market that yields the same optimal strategy and
utility. Such frictionless market is called shadow market for the utility maximization problem. 

\begin{definition} \label{ShadowPriceDef}
A semimartingale $\widetilde{S}=(\widetilde{S}_t)_{0\le t\le T}$ is called \textnormal{shadow price process} 
  for the optimization problem \eqref{J5} if
 \begin{enumerate}[(i)]
   \item $\widetilde{S}$ takes its values in the bid-ask spread $[(1-\la)S, S]$. 
  
   \item     
   The solution $\wt \varphi=(\wt \varphi^0,  \wt \varphi^1)$ to the corresponding frictionless utility maximisation problem 
      $$\massE\big[U\big(x+\varphi^1 \cdot \wt S_T\big)\big] \to max!, \quad (\varphi^0, \varphi^1)\in \cA(x; \wt S)$$ 
     exists and coincides with the solution $\wh \varphi= (\wh \varphi^0, \wh \varphi^1)$ to \eqref{J5} under transaction costs, where $\cA(x; \wt S)$ denotes the set of all self-financing and admissible trading strategies for the price $\widetilde{S}$ without transaction costs in the classical sense as in \cite{KS99}.
 \end{enumerate}
\end{definition}

If an ODP for the dual problem of the utility maximization satisfies appropriate conditions, then a shadow price process can be constructed by this ODP. This is concluded in \cite{CS16duality}.
 
\begin{proposition}[{\cite[Proposition 3.7]{CS16duality}}]  \label{shadow}
For a fixed $x>0$, suppose that all conditions for Theorem \ref{thhurt} hold. 
Assume that the dual optimizer $\widehat{h}(y)$ equals $\widehat{Y}^0_T(y)$, 
where $u'(x)=y$ and $\widehat{Y}(y)\in\mathcal{B}(y; \massP):=(\widehat{Y}^0(y), \widehat{Y}^1(y))$ is a couple of \textnormal{$\massP$-local martingale}.
Then, the strictly positive semimartingale $\widehat{S}:=\frac{\widehat{Y}^1(y)}{\widehat{Y}^0(y)}$ is a {\it shadow price process} 
 for the optimization problem \eqref{J5}.
\end{proposition}

On the other hand, each shadow price process can be represented by the quotient of some ODPs. Here below is the result obtained in \cite{CS16duality}.

\begin{proposition}[{\cite[Proposition 3.8]{CS16duality}}] \label{spstru}
If a shadow price $\wh S$ exists, it is given by $\widehat{S}=\frac{\widehat{Y}^1(y)}{\widehat{Y}^0(y)}$ for an ODP $(\widehat{Y}^0(y),\widehat{Y}^1(y))\in \cB(y,\p)$ of the dual problem \eqref{geht}.
\end{proposition}

\section{Static stability}

For a fixed physical probability measure $\massP$ and a fixed utility function $U$, for any initial endowment $x$, by Theorem \ref{thhurt}, 
the primal problem \eqref{J5} can be solved by some $\wh\varphi^0_T(x; U, \p)\in \mathcal{C}(x; \massP)$ and the optimal value is denoted by $u(x; U, \p)$. On the other hand, for the dual input $y>0$,  the dual problem \eqref{geht} parameterized by $(V, \massP)$ admits a solution $\widehat Y_{T}^0(y; V, \p)\in \mathcal{D}(y; \p)$ and the optimal value is denoted by $v(y; V, \p)$.
In this section, we study how the optimality of the primal and dual problems is affected by
\begin{itemize}
\item perturbations of the initial endowment and of the dual input; 
\item {variations} of the investor's utility; 
\item  misspecification of the underlying market model.
\end{itemize}

Assume that the variations of the investor's utility is represented by a sequence of perturbed utility functions $(U_n)_{n\in\mathbb{N}}$ and that of the underlying market model is described by a sequence of probabilities $(\p_n)_{n\in \mathbb{N}}$. 
Now we introduce the strict mathematical formulation of this problem, in which we make assumptions in accordance with the ones in \cite{KZ11} for considering the frictionless case with random endowments. 

\begin{assumption}\label{pUnconv}
For each $n \in \NN,\ \p_n \sim \p $ and $\lims{n}{\massP} = \massP$ in total variation,
$\lim_{n \to \i}U_n=U$ pointwise and $x_n\rightarrow x>0,\ y_n\rightarrow y>0$.
\end{assumption}

\begin{remark}\label{rem32}
\begin{enumerate}[(1)]
 \item  For each $n$, the Radon-Nikodym derivatives $\frac{d\massP_n}{d\massP}$ exists. Moreover, denote
          $$\wt Z_{t}^{n}:=\Excond{\p}{\frac{d\p_n}{d\p}}{\cF_{t}}, $$
          which is a $\p$-martingale.  
        It is obvious that $\mathcal{Z}^{loc, \lambda}(\massP)\neq \emptyset$ implies $\mathcal{Z}^{loc, \lambda}(\massP_n)\neq \emptyset$, for each $n\in \mathbb{N}$. Precisely, 
          {for any $Z\in \mathcal{Z}^{loc, \lambda}(\massP)$, $\widetilde{Z}^{-1}Z\in \mathcal{Z}^{loc, \lambda}(\massP_n)$ and for any $Z'\in  \mathcal{Z}^{loc, \lambda}(\massP_n)$, $\widetilde{Z}Z'\in \mathcal{Z}^{loc, \lambda}(\massP)$.}
        Moreover, {we define $\widetilde{Z}_n:=\widetilde{Z}^n_T=\frac{d\massP_n}{d\massP}$.}
 \item That $\lim_{n \to \infty}U_n=U$ pointwise implies the pointwise convergence of the sequence of their Legendre-Fenchel transforms, i.e.,  $\lim_{n \to \infty}V_n=V$  pointwisely.  
       Indeed, $\{U_n\}_{n\in\NN}$ is a family of concave functions on finite-dimensional space. 
       Then, pointwise convergence is equivalent to epi-convergence on the interior of the domain of the limiting function and moreover, $U_n$ epi-converges to $U$ is equivalent to the conjugate sequence $V_n$ converges to $V$ (for more general analysis, see \cite{SW77, RW98}). 
 \item Due to the convexity, the sequences $\seqnet{U}{n}{\NN}, \seqnet{V}{n}{\NN}$ as well as their derivatives $\{U'_n\}_{n\in\NN}$, $\{V'_n\}_{n\in\NN}$ converge {\it{uniformly}} on compact subsets of $(0, \infty)$ to their respective limits $U$,  $V$, $U'$ and $V'$ (see  e.g.~\cite[Theorem 10.8 and 25.7]{Roc70} for a general statement).
\end{enumerate}
\end{remark}

\vspace{4mm}

In the frictionless case, Larsen mentioned in \cite[Section 2.6]{Lar09} that the pointwise convergence of $\{U_n\}_{n\in \mathbb{N}}$ is not sufficient to prove the upper semi-continuity of the value function $v$ 
  and more structure conditions on the converging sequence of $\{U_n\}_{n\in \mathbb{N}}$ should be imposed. 
In the present paper, we introduce the following assumption, whose analogue in the frictionless case could be found in \cite[Section 2.3.2]{KZ11}.

\begin{assumption}\label{UIZ}
 Define $$ \cZ_T:= \big\{ Z_T^0: (Z^0,Z^1) \in \cZ^{loc,\lambda} \big\}.$$
 There exists $Z_T^0 \in \cZ_T$, such that for all $y > 0$ the family $\left\{ \wt{Z}_n V_n^+\left(y \frac{Z_T^0}{ \wt{Z}_n}\right)   \right\}_{n \in \NN} $
 is $\p$-uniformly integrable.
\end{assumption}

\begin{theorem} \label{tconv}
 Under the assumptions for Theorem \ref{thhurt} and under Assumptions \ref{pUnconv}, \ref{UIZ}, 
  we have the following limiting relationship for value functions and  optimal solutions:
\begin{align}
  \lim_{n\to \i}  u(x_n; U_n, \p_n)=  u(x; U, \p), &\quad  \lim_{n\to \i}  v(y_n; V_n, \p_n)=  v(y; V, \p);\label{stabvl}\\
  \lim_{n\to \i} \frac{\partial}{\partial x}u(x_n; U_n, \p_n)=   \frac{\partial}{\partial x}u(x; U, \p), &\quad  \lim_{n\to \i} \frac{\partial}{\partial y}v(y_n; V_n, \p_n)=   \frac{\partial}{\partial y}v(y; V, \p);\label{stabde}\\
  \lim_{n\to\i} \wh\varphi^{n, 0}_T(x_n; U_n, \p_n)= \wh\varphi^0_T(x; U, \p),&\quad \lim_{n\to \i}\widehat h(y_n; V_n, \p_n)= \widehat h(y; V, \p),\ \mbox{a.s.}\label{stabop}
\end{align}
\end{theorem}

To prove this theorem, we first need to reformulate the dual problem $v_n$ with the help of the following proposition. 
 \begin{proposition} \label{relaynpn}
If $\widetilde h \in  \cD(y, \p)$, then for each  $n \in \NN$, $\frac{\widetilde h}{\wt Z_{n}}\in  \cD(y, \p_n)$.
       Conversely, if  $h \in  \cD(y, \p_n)$, then $h\wt Z_{n} \in \cD(y, \p).$
        \end{proposition} 
\begin{proof}
  Without loss of generality, let $h \in  \cD(1, \p_n) $.  
  By definition of $\cD(1, \p_n)$, there is a nonnegative $\p_n$-optional strong supermartingale $ (Y^0,Y^1)$ with $Y_0^0= 1,\ Y_T^0= h$, 
    satisfying that $\frac{Y^1}{Y^0} \in [(1-\la)S, S], \mbox{ and } \varphi^0Y^0+ \varphi^1Y^1$ is a nonnegative $\p_n$-optional strong supermartingale for all $\varphi \in \cA(1)$. 
  Define 
   $$\wt Z_{t}^{n}:=\Excond{\p}{\frac{d\p_n}{d\p}}{\cF_{t}},$$
   which is a $\p$-martingale. 
  Clearly, $\wt Z_n = \wt Z_{T}^{n}$.

  It suffices to show the supermartingale property of the process $ \varphi^0Y^0\wt Z^{n}+ \varphi^1Y^1 \wt Z^{n}$ under $\p$, for all $\varphi \in \cA(1)$.
  For $0  \leq s< t \leq T$   and  $\varphi=(\varphi^0_t, \varphi^1_t)_{\tT} \in \cA(1)$, by Bayes' formula, we have immediately 
    $$( \varphi^0_s Y^0_s +  \varphi^1_s Y^1_s )\wt Z_{s}^{n} \geq  \Excond{\p}{\left(\varphi^0_t Y^0_t +  \varphi^1_t Y^1_t\right)\wt Z_{t}^{n}}{\cF_{s}}.$$ 
 
  Analogously we may show the first assertion. 
\end{proof}

\begin{proof}[Sketch of proof of Theorem \ref{tconv}] The proof, which is very similar as in  \cite{KZ11} by setting the random endowments $q=0$,  will be basically in four steps. In order to avoid redundancy, we just sketch the proof and list what will be needed in the context with transaction costs. We therefore refer the readers to \cite[Section 3]{KZ11} for more details. 

\begin{enumerate}[(1)]
 \item 
By Proposition  {\ref{relaynpn}}, the dual problem $v(\cdot; V_n, \p_n)$ corresponding to the utility maximization problem parameterized by $(U_n, \p_n)$ can be  reformulated with the parameters  $(U_n, \p)$. Precisely,
  \begin{align}\notag
    v(y; V_n, \p_n) &= \inf_{h \in \cD(y,\p_n)} \Ex{\p_n}{V_n(h)}=\Ex{\p_n}{V_n\left(\wh{h}(y,\p_n)\right)}\notag\\
                      &= \inf_{\overline{h} \in \cD(y,\p)} \Ex{\p}{ \widetilde{Z}_nV_n\left(\frac{\overline h}{\widetilde{Z}_n }\right)}=  \Ex{\p}{ \widetilde{Z}_nV_n\left(\frac{\wh h_n(y,\p)}{\widetilde{Z}_n }\right)}, \notag
    \end{align}
   where $\wt{Z}_n$ converges to 1 in $\lo$ and thus in $\lzs(\p)$; the dual optimizer 
    $\wh{h}(y; \p_n) \in \cD(y, \p_n)$ and $\wh{h}_n(y; \p) \in \cD(y, \p)$ attain the infimum $v(y; V_n, \p_n)$.
  The equality above corresponds to (3.1) in \cite{KZ11}. 
  We can subsequently proceed exactly the same as \cite[Section 3.2]{KZ11} to have the semi-continuity
    $$v(y; V, \p)\leq \liminf_{n\rightarrow \infty}v(y; V_n, \p_n),\quad {\rm for}\ y>0, $$
    the proof of which depends on the properties of the functions $\{V_n\}_{n\in \mathbb{N}}$ and $V$ {as well as} Assumption \ref{pUnconv} and Remark \ref{rem32}.

 \item The dual value function has an upper-semicontinuity property. Indeed,
  for any  $Z_T^0 \in  \cZ_T, y > 0$, we define 
       $$\wtcD(y, Z_T^0):= \left\{ h \in \cD(y,\p):  \frac{Z_T}{h} \in L^{\infty} \right\},$$
       then the counterparts of Lemma 3.3 and 3.4 from \cite[Section 3.3]{KZ11} in the context with transaction costs are listed as follows:      
        
        \begin{lemma} [compare {\cite[Lemma 3.3]{KZ11}}] \label{vexp}
         Fix $ y >0 $, let   $Z_T^0 \in  \cZ_T$ be such that $V^+(yZ_T^0)\in \lo$,
         Then $$v(y; V, \p)= \inf_{\wt h \in \wtcD(y, Z_T^0)}\Ex{\p}{V\big(\wt h\big)}.$$
        \end{lemma}

       \begin{lemma}[compare {\cite[Lemma 3.4]{KZ11}}]
       Suppose for some $f \in \lzs_+,$ the collection $\{\widetilde{Z}_n V^+_n(f/\widetilde{Z}_n)\}_{n\in \NN}$ is $\p$-uniformly integrable.
       Let $h \in \lo$ be such that $h \geq f$ a.s. Then
       $$\lim_{n \to \i}\widetilde{Z}_n V_n\big(h/\widetilde{Z}_n\big) =V(h) \mbox{ in\ }\lo.$$
       \end{lemma}
       With the help of these two lemmas, we can show as \cite[Section 3.3.2]{KZ11}  that 
       $$v(y; V, \p)\geq \limsup_{n\rightarrow \infty}v(y; V_n, \p_n), \quad {\rm for}\ y>0.$$

\item From the results in (1) and (2) and the fact that $v(\cdot; V_n, \p_n)$ and $v(\cdot; V_n, \p_n)$ are convex functions, we can conclude that
       $$ \lim_{n\rightarrow \infty }v(y_n; V_n, \p_n)=v(y; V, \p). $$
On the other hand, since $u(\cdot; U, \p)$ (resp. $u(\cdot; U_n, \p_n)$) is the Legendre-Fenchel transform of $v(\cdot; V, \p)$ (resp. $v(\cdot; V_n, \p_n)$). For the similar reason as Remark \ref{rem32} (2) and (3), we have 
     $$ \lim_{n\rightarrow\infty}u(x_n; U_n, \p_n)= u(x; U, \p). $$
Thus, (\ref{stabvl}) is proved. To see (\ref{stabde}), it suffices to apply the epi-convergence properties of $u(\cdot; U_n, \p_n)\rightarrow u(\cdot; U, \p)$ and of $v(\cdot; V_n, \p_n)\rightarrow v(\cdot; V, \p)$ and proceed the argument on the graphical convergence of the subdifferentials, which is the same as \cite[Section 3.4]{KZ11}.
 
\item Finally, to prove the $L^0$-convergence of the dual optimizer $\lim_{n\to \i}\widehat h(y_n; V_n, \p_n)= \widehat h(y; V, \p)$ in (\ref{stabop}), we can first define an auxiliary sequence $f_n:=n^{-1}yZ^{0}_T+(1-n^{-1})\widehat h(y; V, \p)\in \wt \cD(y, Z_T^0)$, for some $Z^0_T\in \mathcal{Z}_T$ such that $V^+(Z^0_T)\in L^1(\massP)$. Obviously, $f_n\rightarrow \widehat h(y; V, \p)$ in $L^0$. Then, one could prove by  following the procedures in \cite[Section 3.5.2]{KZ11} the claim that there exists a subsequence indexed by ${\{n_k\}}_{k\in \mathbb{N}}$, such that 
 $$
 \lim_{k\rightarrow \infty}\p_{n_k}\left\{ \widehat h(y_{n_k}; V_{n_k}, \p_{n_k})\in [k^{-1}, k], \frac{f_{n_k}}{\widetilde{Z}^{n_k}_T}\in [k^{-1}, k],  \left|\widehat h(y_{n_k}; V_{n_k}, \p_{n_k})- \frac{f_{n_k}}{\widetilde{Z}^{n_k}_T}\right|>k^{-1}\right\}=0,
 $$
where a more elaborate version of the method used in the proof of \cite[Lemma A1.1]{DS94} plays the key role. The other part of (\ref{stabop}) can be obtained in view of (\ref{Y=U'(X)}), (\ref{stabde}) and Remark \ref{rem32}.
\end{enumerate}
\end{proof}

\section{Dynamic stability} 
In the previous section, the stability result is only on the terminal values of the wealth processes as well as the terminal values of the optional supermartingale deflators attaining the dual optimality.  In the present section, we focus on the dynamics of the whole process which attains the optimizer.
 In particular, we will first investigate properties of the optimal dual processes and discuss the so-called dynamic stability. Moreover, we construct a shadow price process for the limiting problem parameterized by $(x; U, \p)$ by the sequence of shadow prices corresponding to the problems parameterized by $\{(x_n; U_n, \p_n)\}_{n\in \mathbb{N}}$.\\

Throughout this section, we need the following assumption, which is necessary for the existence of shadow price processes (see {\cite{CSY17}}).

\begin{assumption}\label{Scont}
Additionally, we assume that the stock price is driven by a continuous process.
\end{assumption}

\subsection{Properties of optional strong supermartingale deflators}
In this subsection, we first study the properties of optional strong supermartingale deflators, which are necessary for the argument on the dynamic stability. For the simplicity of notation, we consider only $Y:=(Y^0, Y^1)\in \mathcal B(1)$. Since the processes $Y^0$ and $Y^1$ are optional strong supermartingales, according to ${\cite{Mer72}}$, they admit the {Doob-Meyer-Mertens} decomposition, which is an analogue of the Doob-Meyer type decomposition for c\`adl\`ag supermartingales, i.e.,
\begin{align}\label{Mdecom}
 Y^i = M^i -A^i, \quad i=0, 1,\end{align}
where $M^i$ is a c\`adl\`ag local martingale and $A^i$ is a nondecreasing l\`adl\`ag predictable process. We now introduce the following proposition, which shows the relation between the nondecreasing parts $A^0$ and $A^1$, that is, symbolically,
   \begin{equation}\notag
      (1-\lambda)S_tdA_t^0 \leq dA_t^1 \leq S_tdA_t^0.
    \end{equation}
\begin{proposition} \label{propA}
  Let Assumption \ref{Scont} hold and  $Y:=(Y^0,Y^1)\in\cB(1)$.  The processes $Y^0$ and $Y^1$ admit the unique Doob-Meyer-Mertens decomposition as (\ref{Mdecom}).
  Let $\varepsilon+\lambda\leq 1$ and let $\sigma$ be a stopping time taking values in $[0, T]$. 
  Define 
     $$ \tau_{\varepsilon}:=\inf\left\{t\geq\sigma \,\Big|\,\frac{S_t}{S_\sigma}=(1+\varepsilon)\mbox{ or } (1-\varepsilon)\right\}\wedge T. $$
  Then, for all stopping time $\tau$ satisfying $\sigma\leq\tau\leq \tau_{\varepsilon}$, 
    \begin{equation}\label{a0a1}
     \begin{aligned}
         (1-\varepsilon)(1-\lambda)S_\sigma\massE\left[A^0_{\tau}-A_\sigma^0\big|\cF_\sigma\right] &\leq \massE\left[A^1_{\tau}-A^1_\sigma\big|\cF_\sigma\right]  \\
            &\leq (1+\varepsilon)S_\sigma\massE\left[A^0_{\tau}-A_\sigma^0\big|\cF_\sigma\right].
     \end{aligned}
    \end{equation}
\end{proposition}
 
\begin{remark}
 A similar lemma can be found in \cite[Lemma 3.5]{CSY17}, which gives the property (\ref{a0a1}) for a particular supermartingale deflator constructed as the Fatou limiting process of a sequence of local consistent price systems. 
 We remark that our simple proof of the above proposition will not depend on the construction of $Y$, but only on the supermartingale property of the process $\varphi^0Y^0+\varphi^1Y^1$, for all $\varphi\in \mathcal{A}(1)$.
\end{remark}

\begin{proof}
 Consider the following  trading strategy $\varphi=(\varphi^0_t,\varphi^1_t)_{0\leq t\leq T} \in \cA(1)$ defined by 
    \begin{equation*}
       (\varphi^0_t,\varphi^1_t) := \left\{\begin{array}{cc} 
                                             \left(1-(1-\lambda)(1-\varepsilon), 0\right),  &  0\leq t <\sigma;  \\
                                             \left(-(1-\lambda)(1-\varepsilon),\frac{1}{S_\sigma}\right),  &  \sigma\leq t< \tau_\varepsilon;  \\
                                             \left(V^{liq}_{\tau_\varepsilon}(\varphi), 0 \right), & \tau_\varepsilon \leq t\leq T. 
                                           \end{array}
         \right.
    \end{equation*}
 Obviously, this trading strategy defined as above is $\lambda$-self-financing and admissible. Indeed, we only need to consider the worst case, i.e., the stock price attains $(1-\varepsilon)S_\sigma$,
     then we have for $\sigma < t \leq T$, the liquidation value
     $$ V^{liq}_t(\varphi) \geq -(1-\lambda)(1-\varepsilon) + (1-\lambda)\frac{1}{S_\sigma}(1-\varepsilon)S_\sigma = 0.  $$
 By the definition of the optional strong supermartingale deflator, we obtain that for each $\tau\in \llbracket \sigma, \tau_{\varepsilon}\rrbracket$,
    \begin{equation}\notag
     \begin{aligned}
       & -(1-\lambda)(1-\varepsilon)(M^0_\sigma-A^0_\sigma) + \frac{1}{S_\sigma}(M^1_\sigma-A^1_\sigma) \\
       & \quad \geq \massE\left[-(1-\lambda)(1-\varepsilon)\left(M^0_{\tau}-A^0_{\tau}\right)+\frac{1}{S_\sigma}\left(M^1_{\tau}-A^1_{\tau}\right)\Big|\cF_\sigma\right]. 
     \end{aligned}
    \end{equation}
 We may assume without loss of generality that $M^0$ and $M^1$ are true martingales, otherwise, stopping technique applies.
 Therefore, we obtain
    \begin{equation}\notag
      \massE\left[A^1_{\tau}-A^1_\sigma|\cF_\sigma\right] \geq (1-\varepsilon)(1-\lambda)S_\sigma\massE\left[A^0_{\tau}-A^0_\sigma\big|\cF_\sigma\right].
    \end{equation}
    
 For the other inequality, let us define another trading strategy $\varphi=(\varphi^0_t,\varphi^1_t)_{0\leq t\leq T} \in \cA(1)$ by 
    \begin{equation*}
       (\varphi^0_t,\varphi^1_t) := \left\{\begin{array}{cc} 
                                             \left(\lambda+\varepsilon, 0\right),  &  0\leq t <\sigma;  \\
                                             \left((1-\lambda)+\lambda+\varepsilon,-\frac{1}{S_\sigma}\right),  &  \sigma\leq t< \tau_\varepsilon;  \\
                                             \left(V^{liq}_{\tau_\varepsilon}(\varphi), 0 \right), & \tau_\varepsilon\leq t\leq T. 
                                           \end{array}
         \right.
    \end{equation*}     
  With the same argument, we have that for each $\sigma\leq \tau\leq \tau_{\varepsilon}$,
    \begin{equation}\notag
      (1+\varepsilon)S_\sigma\massE\left[A^0_{\tau}-A^0_\sigma\big|\cF_\sigma\right] \geq \massE\left[A^1_{\tau}-A^1_\sigma|\cF_\sigma\right], 
    \end{equation}
    which ends the proof.
\end{proof}

The following corollary is straightforward.
\begin{corollary}\label{slm}
 For each optional strong supermartingale deflator $\big({Y}^0, {Y}^1\big)\in\mathcal{B}(1)$, if the first component ${Y}^0$ is a local martingale, then the second component  ${Y}^1$ is also a local martingale.

\end{corollary}

\begin{remark}
 Actually,  we can still have Corollary \ref{slm}  without Lemma \ref{propA}. By defining an admissible  trading strategy $({\varphi}^0_t,{\varphi}^1_t)_{0 < t \leq T}=(1+\frac{1}{n},-\frac{1}{n})$ and  $({\varphi}^0_t,{\varphi}^1_t)_{0 < t \leq T}=(1-\frac{1}{n},\frac{1}{n}), n \in \NN$,  we can complete the proof.
 \end{remark}

Next, we study the local martingale property of the ODP in the framework of \cite{CSY17}.
Actually, Czichowsky et al. found in \cite{CSY17} a sufficient condition which guarantees this local martingale property of ODPs when $S$ is continuous, 
  that is, the liquidation value process of OPPs is a.s.~strictly positive, i.e., $\inf_{\tT} {\widehat V}^{liq}_t(\widehat \varphi) > 0$, a.s. 
We summarize this result from \cite{CSY17} as the following proposition.

\begin{proposition}[{\cite[Proposition 3.3]{CSY17}}]\label{doplm}
 Fix the level $0<\la <1$ of transaction costs. 
 We assume that Assumption \ref{Scont} and all assumptions for Theorem \ref{thhurt} are satisfied. 
 In addition, we suppose that the liquidation value process of the optimal primal process $\widehat \varphi = (\wh \varphi^0_t,\wh \varphi^1_t)_{\tT} \in \cA(x)$ is strictly positive, i.e.,
  \begin{equation}  \label{spodp}
   \inf_{\tT}\widehat V^{liq}_t(\widehat \varphi):=  \inf_{\tT} \left\{ \widehat \varphi^0_t + (1- \la) (\wh \varphi^1_t)^+ S_t - (\wh \varphi^1_t)^- S_t \right\} > 0, \ a.s.
  \end{equation}
 Let $y=u'(x)$,  then there exists a local $\lambda$-consistent price system $\widehat{Z}\in \mathcal{Z}^{loc, \lambda}$, such that $y\widehat Z^0_T=\widehat{h}$.
\end{proposition}

\begin{remark}
 The continuity of the price process $S$ is essential in the above proposition. Otherwise, counterexamples could be found, e.g., in \cite{CMKS14}.
 Very recently, it is proved in the paper \cite{CPSY16} that \eqref{spodp} holds if we assume the condition of ``two way crossing" (TWC) (see Bender \cite{Ben12} for definition).
This condition implies that $S$ satisfies ($CPS^{\mu}$) locally for all $0 < \mu < 1$. (See the proof of Theorem 2.3 in \cite{CPSY16}.)  
\end{remark}

The proof of the proposition above in \cite{CSY17} is based on the strict positivity of the liquidation value of OPP and \cite[Lemma 3.5]{CSY17} (see also Lemma \ref{propA}). 
In the present paper, we give a proposition which is a little stronger than \cite[Proposition 3.3]{CSY17}, however, its proof is much reduced. 
Indeed, the result of Proposition \ref{doplm} constructs such ODP which is a couple of local martingales and this is sufficient for the construction of a shadow price process, 
    whereas the proposition below states that all ODPs are local martingales in the same framework. 
This statement is meaningful in a frictional context, since unlike the frictionless case, the ODP is no longer unique (a counterexample will be discussed later).

\begin{proposition}\label{doplm1}
 Fix the level $0<\la <1$ of transaction costs. 
 Under the assumptions of Proposition \ref{doplm}, let $y=u'(x)$ and $\widehat{Y}\in \mathcal{B}(y)$ be the optional supermartingale deflator associated with the dual optimizer in the sense that $\widehat{Y}^0_T=\widehat{h}$ a.s. 
 Then, $\widehat{Y}$ is a local {martingale.} 
 In other words, all $\lambda$-optional supermartingale deflator associated with the dual optimizer is a local $\lambda$-consistent price system.
\end{proposition}

\begin{proof} 
The stock price process $S$ is continuous and thus $S$ is locally bounded. 
 We may assume without loss of generality that the liquidation value of the OPP $(\widehat V_t^{liq})_{0 \leq t \leq T}$ is predictable (otherwise, consider the c\`agl\`ad version of the trading strategy), 
   then from $\inf_{\tT}\wh V_t^{liq} > 0$, we can find a sequence of a.s.~increasing and diverging stopping times $\sequ{\varrho}{m}$, 
   such that, for each $m\in \mathbb{N}$ and $t\in[0,T]$, we have $\widehat V^{liq}_{t \wedge \varrho_m}\geq \frac{1}{m}$,
    $$\tfrac{1}{m} \leq (1- \lambda)S_{t \wedge \varrho_m} \leq S_{t \wedge \varrho_m}.$$
 Fix the optimal trading strategy $\big(\widehat{\varphi}^0(x),\widehat{\varphi}^1(x)\big)\in\mathcal{A}(x)$. 
 It is easy to verify that $\big(\widehat{\varphi}^0_{\cdot\wedge \varrho_m}-\frac{1}{m},\widehat{\varphi}^1_{\cdot\wedge \varrho_m}\big)\in\mathcal{A}(x)$. 
 Thus, 
  $$ \left(\widehat{\varphi}^0_{\cdot\wedge \varrho_m}-\frac{1}{m}\right)\widehat{Y}^0+\widehat{\varphi}^1_{\cdot\wedge \varrho_m}\widehat{Y}^1 $$
    is an optional strong supermartingale. 
 Recall that $\widehat \varphi^0\widehat{Y}^0+\widehat{\varphi}^1\widehat{Y}^1$ is a martingale, so that $\widehat{Y}^0$ is an optional strong submartingale up to $\varrho_m$. 
 As $\widehat{Y}\in \mathcal{B}(y)$, $\widehat{Y}^0_{\cdot\wedge \varrho_m}$ has to be a martingale. So $\widehat{Y}^0$ is a local martingale and by Corollary \ref{slm}, $\widehat{Y}^1$ is again a local martingale.
\end{proof}

\begin{remark}
Similar to the frictionless case, an ODP is always a couple of local martingales when $S$ is continuous. 
It is to say that whenever the liquidation value is strictly positive, the first coordinate of the ODP will not ``lose'' its mass. 
This can be seen if we proceed a similar argument as the \cite[Appendix]{GLY16} to prove the above proposition. 
It suffices to replace the wealth process by the  liquidation value process.
\end{remark}

Unlike the frictionless case, we do NOT have the uniqueness for the ODPs, even for its first coordinate. 
The non-uniqueness of the second coordinate ODPs is an easy observation. 
It may already occur in the setting of finite $\Omega$ as was observed in \cite[Example 2.4]{Sch17}. 
For the non-uniqueness of the first coordinate of ODPs, we gratefully acknowledge Walter Schachermayer for the following counterexample (see also \cite{Sch17}).
 
\begin{example}[Walter Schachermayer] 
Let $W=(W_t)_{t\geq 0}$ be an $(\cF^W_t)_{t\geq 0}$-Brownian motion, where $(\cF^W_t)_{t\geq 0}$ is the natural filtration generated by $W$ satisfying the  usual conditions. 
We assume that the investor's utility is characterized by a logarithmic function $U(x)=\log(x)$.  
In what follows, we first construct two frictionless markets driven by two different processes $\widehat{S}$ and $\check{S}$, and eventually find a market $S$ with transaction costs, such that $\widehat{S}$ and $\check{S}$ are both shadow price processes for this logarithmic utility maximization problem when trading with $S$. 
Moreover, we prove that $\widehat{Z}=(\widehat{S}^{-1}, 1)$ and $\check{Z}=(\check{S}^{-1}, 1)$ are two ODPs for such problem, but $\widehat{Z}\neq \check{Z}$. 
The construction is divided into three steps: 

\vspace{3mm}

\noindent {\bf Step 1:}
Define
$$N_t=\exp (W_t +\tfrac{t}{2}), \qquad t\geq 0.$$
Fix the level of transaction costs $0<\lambda<\frac{1}{2}$ and define
 $$ \tau^\lambda =\inf\left\{t: \ N_t=2(1-\lambda)\right\}. $$
Let $\wh S$ be a time-changed restriction of $N$ on the stochastic interval $\llbracket 0, \tau^\lambda\rrbracket$, i.e.,
 $$\wh{S}_t = N_{\tan\left(\frac{\pi}{2}t\right) \wedge \tau^\lambda}, \qquad 0 \leq t \leq 1.$$
Consider a frictionless market driven by the process $\widehat{S}$, which is adapted to the time-changed filtration $\mathcal{F}$ defined by $\mathcal{F}_t:=\mathcal{F}^W_{{\tan\left(\frac{\pi}{2}t\right) \wedge \tau^\lambda}}$. 
Then, $\wh{Z}^0:=(\wh{S})^{-1}$ defines a local martingale deflator for $\widehat{S}$. 
Due to Merton's rule, the logarithmic utility maximization problem is solved by the strategy consisting of buying one stock at time $t=0$ at price $\wh{S}_0=1$ and selling it at time $t=1$ at price $\wh{S}_1=2(1-\lambda)$ 
   (we call it buy-hold-sell strategy). 

\vspace{3mm}
   
\noindent {\bf Step 2:} Next we define a perturbation of the process $\wh{S}$ which will be denoted by $\check{S}.$ 
To do so we first define a perturbation of $\wh{Z}^0=(\wh{S})^{-1},$ then  decompose this local martingale into $\wh{Z}^0=\wh{M}+ \wh{P}$, where 
   $$ \wh{M}_t := \mathbb{E}\Big[\wh{Z}^0_1\Big|\cF_t\Big]=\frac{1}{2(1-\lambda)}, \qquad \wh{P}_t:=\wh{Z}_t-\wh{M}_t. $$
Note  that $\wh{P}_0=\frac{1-2\lambda}{2(1-\lambda)}$ and $\wh{P}_1=0.$ 
Therefore, $\wh P$ is a potential. 
Let $\sigma:=\inf \big\{t: \wh{P}_t=1\big\}$, then the stopped local martingale  $\wh{P}^{\sigma}$ is bounded and thus is a martingale. 
Moreover, $\p[\sigma < \infty]= \frac{1-2\lambda}{2(1-\lambda)}$. 
For $\delta>0$ choose an arbitrary $\mathcal{F}_{\sigma}$-measurable function $f$ taking values in $[1-\delta, 1+\delta]$ such that 
   $$ \mathbb{E}\big[f \mathbf{1}_{\{\sigma < \infty\}}\big]=\frac{1-2\lambda}{2(1-\lambda)} $$ 
   and such that $f$ is not identically equal to $1$ on $\{\sigma<\infty\}$. 
Define the potential $\check{P}$ by 
\begin{align}\notag
\check{P}_t = \left\{ 
   \begin{array}{cl}
     \mathbb{E}[f \mathbf{1}_{\{\sigma < \infty\}} \,|\, \mathcal{F}_{t \wedge \sigma}],  \quad & 0 \leq t \leq \sigma,\\
     f \wh{P}_t, & \sigma\leq t\leq 1,
   \end{array}\right.
\end{align}
 which is again a local martingale starting at $\check{P}_0=\frac{1-2\lambda}{2(1-\lambda)}$ and ending at $\check{P}_1=0.$ Note that $0\leq {\check{P}_t}\in [(1-\delta){\wh{P}_t}, (1+\delta){\wh{P}_t}],$ $0 \leq t < 1$, a.s. 
Define $\check{Z}^0:=\wh{M}+ \check{P}$ and $\check{S}:=(\check{Z}^0)^{-1}.$  
Then, the ratio $\frac{\check{S}_t}{\wh{S}_t} \in [(1+\delta)^{-1}, (1-\delta)^{-1}].$  
It is clear that the frictionless market driven by $\check{S}$ has a local martingale deflator $\check{Z}$. 
This market again has the property that the log-optimal strategy is the buy-hold-sell one as before.

\vspace{3mm}

\noindent{\bf Step 3:} Define 
$$m_t=\max (\wh{S}_t, \check{S}_t), \quad M_t=(1-\lambda)^{-1} \min(\wh{S}_t, \check{S}_t).$$
It suffices to assume that $(1-\lambda)(1+\delta) < (1-\delta)$  to have $m_t < M_t$, $0 \leq t \leq 1$, a.s. Define $S$ by
  $$ S_t=(1-t)m_t + t M_t, \qquad 0 \leq t \leq 1, $$
  which is continuous and adapted starting at $S_0=1$ and ending at $S_1=2$. 
Obviously, both $\wh S$ and $\check S$ remain in the bid-ask spread $[(1-\lambda)S,S].$ 
Therefore, it is clear that both $(\wh Z^0, 1)$ and $(\check Z^0, 1)$ are elements in $\cZ^{loc, \lambda}(S)$. 
Since trading for $S$ with transaction costs yields no more than trading for $\wh S$ or $\check S$ frictionlessly, the trading strategy to buy one stock at $t=0$ and to sell it only when $t=1$ is optimal for the frictional logarithmic utility maximization problem with $S$. 
By \eqref{Y=U'(X)}, we can verify that both $(\wh Z^0, 1)$ and $(\check Z^0, 1)$ induce the dual optimizer $\wh{h}=\frac{1}{2(1-\lambda)}$. However, $\wh Z^0\neq \check Z^0$.
\end{example}

\subsection{Convergence of optimal dual processes}

 In this subsection, we move to discuss the dynamic stability, which means the stability of the ODPs. 
 For the sake of simplicity, we first consider the utility maximization problem under perturbations of the initial endowment and of the utility function, while we keep the physical probability measure unchanged, i.e., $\p_n \equiv \p$. 
 We assume all conditions of Proposition \ref{doplm} for each utility maximization problem $u(x_n; U_n)$ as well as for the limiting problem $u(x; U)$. From the result in the previous subsection the optimal dual process $\wh Y^n:=(\widehat Y^{n,0}, \widehat Y^{n,1}):=(\widehat Y^{0}(y_n; V_n), \widehat Y^{1}(y_n; V_n))$ is a couple of $\p$-local martingales, i.e.,\ $\wh Y^n\in y_n\mathcal Z^{loc, \lambda}(\massP)$.

\begin{proposition}\label{nopdconv}
We assume all conditions of Proposition \ref{doplm} for the market model and for the utility function in each utility maximization problem $u(x_n; U_n, \massP)$ (for short, $u(x_n; U_n)$) as well as  for the limiting problem $u(x; U, \massP)$ (for short, $u(x; U)$). Moreover, assume that $\lim_{n \to \i}U_n=U$ pointwise and $x_n\rightarrow x>0$.
Then, for a sequence of ODPs $\wh Y^n:=(\widehat Y^{n,0}, \widehat Y^{n,1})$ that associated with the dual optimizers $\wh h^n$ (i.e., $\wh Y^{n, 0}_T=\wh h^n$) corresponding to the dual problems $v(y_n; V_n)$, 
   there exists a couple of local martingales $\wh Y:=(\widehat Y^{0}, \widehat Y^{1})$ and a subsequence of convex combinations of $\{\wh Y^n\}_{n=0}^\infty$, denoted still by $\{\wh Y^n\}_{n=0}^\infty$, 
   such that for every $[0,T]$-valued stopping time $\tau$,
     $$\big(\widehat{Y}^{n, 0}_{\tau},\widehat{Y}^{n, 1}_{\tau}\big)\stackrel{\massP}{\longrightarrow}\big(\widehat{Y}^{0}_{\tau},\widehat{Y}^{1}_{\tau}\big).$$
  Moreover, the couple $\wh Y\in y\mathcal{Z}^{loc, \lambda}$ is an ODP corresponding to the limiting dual problem $v(y; V)$ with $\wh Y^0_T=\wh h$.
\end{proposition}

\begin{proof} 
 Since for each $n$, $\wh Y^n\in y_n \mathcal{Z}^{loc, \lambda}\subseteq \mathcal{B}(y_n)$, 
 then by \cite[Theorem 2.7]{CS16strong}, there exists a subsequence of convex combinations of  $\seq{\widehat Y}{n}$, still denoted by $\seq{\widehat Y}{n}$
   and an optional strong supermartingale $\widehat Y=(\widehat Y^{0},\widehat Y^{1})$, for every stopping time $\tau$ taking values in $[0,T]$,
    \begin{align}  \label{conver}
       (\widehat Y^{n,0}_{\tau},\widehat Y^{n,1}_{\tau})  \stackrel{\massP}{\longrightarrow} (\widehat Y^{0}_{\tau},\widehat Y^{1}_{\tau}).
    \end{align}
 From the result of Theorem \ref{tconv}, in particular \eqref{stabde} and (\ref{stabop}), we have $\wh Y^{n,0}_T=\wh h^n\stackrel{\massP}{\longrightarrow} \wh h=\wh Y^0_T$. 
 It is obvious that $\wh Y\in \mathcal{B}(y)$ and thus it is an ODP corresponding to the utility maximization problem $u(x; U)$.
 By applying Proposition \ref{doplm1}, the constructed couple $\widehat Y=(\widehat Y^{0},\widehat Y^{1})$ is a local martingale.
\end{proof}

\begin{remark}
 If we assume only the conditions in Theorem \ref{thhurt} for each problem $u(x_n; U_n)$ and the limiting one $u(x; U)$, then the results in the above proposition hold without the local martingale property.  
\end{remark}

\begin{remark}\label{rem1}
For the more general case including perturbations on the physical probability measure, we assume the same as in the above proposition. In addition, we suppose that for each $n\in \NN$,  $\p_n \sim \p $ and $\lims{n}{\massP} = \massP$ in total variation.
For each $n$, we solve a utility maximization problem under $\p_n$, and pick up an ODP $\big(\widehat{Y}^{0}(y_n; V_n, \p_n),\widehat{Y}^{1}(y_n; V_n,\p_n)\big)\in y_n\mathcal{Z}^{loc,\lambda}(\p_n)$ which is a couple of $\p_n$-local martingales. 
By change of measure, we obtain 
   \begin{align}(\widetilde{Y}^{n, 0}, \widetilde{Y}^{n,1})&:=\left({\widehat{Y}^{0}(y_n; V_n, \p_n)} \widetilde{Z}^n, {\widehat{Y}^{1}(y_n; V_n, \p_n)}
      \widetilde{Z}^n \right)
\in y_n\mathcal{Z}^{loc,\lambda}(\p)\subseteq\mathcal{B}(y_n, \p).\label{def1}\end{align}
For each $n$, the processes $(\widetilde{Y}^{n, 0}, \widetilde{Y}^{n,1})$ are $\p$-local martingales.
Moreover, from the fact $\widetilde{Z}^n_T \to 1$ in $L^1(\p)$ and Theorem \ref{tconv}, we have 
                 $$\widehat{Y}^{0}_T(y_n; V_n, \p_n)  \frac{d\p_n}{d \p} \stackrel{\massP}{\longrightarrow} \wh h(y, V, \p)=\widehat{Y}^{0}_T(y; V, \p),$$
                 where the right-hand side is the dual optimizer for $v(y, V, \p)$.
Thus, again by \cite[Theorem 2.7]{CS16strong}, there exists a subsequence of convex combinations of $\{(\widetilde{Y}^{n, 0},\widetilde{Y}^{n, 1})\}_{n\in \mathbb{N}}$, denoted also by $\{(\widetilde{Y}^{n, 0},\widetilde{Y}^{n, 1})\}_{n\in \mathbb{N}}$, such that for every stopping time $\tau$ taking values in $[0,T]$,
    \begin{align}\label{sen1}\big(\widetilde{Y}^{n, 0}_{\tau},\widetilde{Y}^{n, 1}_{\tau}\big)\stackrel{\massP}{\longrightarrow}\big(\widehat{Y}^{0}_{\tau},\widehat{Y}^{1}_{\tau}\big).
    \end{align}
Any such limiting process $\big(\widehat{Y}^{0},\widehat{Y}^{1}\big) \in y\cZ^{loc, \lambda}(\p)$ is an ODP for $v(y; V,\p)$, by applying Proposition \ref{doplm1}, which is a couple of $\p$-local martingale.
\end{remark}

\begin{remark}
  Recall that we assume the usual conditions for the filtration throughout this paper, then every ODP discussed in this subsection has a c\`adl\`ag modification. This  property is crucial for the construction of a shadow market.
\end{remark}

\subsection{Construction of shadow price processes}

This subsection is devoted to the construction of shadow price processes for the limiting utility maximization problem by a sequence of shadow price processes corresponding to the problem with perturbations. 
Similar to the previous subsection, we first consider the case when $\p_n \equiv \p$.

\begin{proposition}\label{nopspconv} 
We assume all conditions in Proposition \ref{nopdconv} for each utility maximization problem $u(x_n; U_n)$ as well as for the limiting problem $u(x; U)$. 
For each $n$, suppose that $\widehat S^n = \widehat S(y_n; V_n)$ is a shadow price process corresponding to the $n$-th problem, which can be represented by $\widehat S^n = \frac{\widehat Y^{n, 1}}{\widehat Y^{n, 0}}$, 
  where $\wh Y^n:=\big(\wh Y^{n,0}, \wh Y^{n,1}\big)\in \mathcal{B}(y_n)$ is a $\p$-local martingale. 
Then, there exists a subsequence of convex combinations of $\big\{\wh Y^n\big\}_{n=0}^\infty$ and a limiting process $\wh Y\in \mathcal{B}(y)$ such that $\wh Y^n$ converges to $\wh Y$ in the sense of \eqref{conver}. 
Moreover, $\widehat S := \frac{\widehat Y^{1}}{\widehat Y^{0}}$ defines a shadow price process for the limiting problem $u(x; U)$.
\end{proposition} 

\begin{proof}
  This proposition can be proved by simply applying  Proposition \ref{nopdconv} and Proposition \ref{spstru}.
\end{proof}

\begin{remark}
For the more complicated situation with perturbations on the physical probability measure, we turn to the setting of Remark \ref{rem1}.
By Proposition \ref{spstru}, for each $n$, there exists a shadow price process $\wh S^n$ for the utility maximization problem $u(x_n; U_n, \p_n)$, which can be represented by some  ODP corresponding to this problem, i.e., 
 $$ \widehat S^n(x_n; U_n, \p_n)=\frac{\widehat{Y}^{1}(y_n; V_n, \p_n)} {\widehat{Y}^{0}(y_n; V_n, \p_n)}  = \frac{\wt{Y}^{1}(y; V_n, \p)} {\wt{Y}^{0}(y; V_n, \p)}, $$
where the sequence $\{(\widetilde{Y}^{n, 0}, \widetilde{Y}^{n,1})\}_{n\in \mathbb{N}}$ is from \eqref{def1}.
Then, any limiting process $(\wh{Y}^{0},\wh{Y}^{1})$ constructed as in Remark \ref{rem1} in the sense \eqref{sen1} defines a shadow price process $\widehat{S}(x; U, \p):=\frac{\wh Y^1}{\wh Y^0}$ for the limiting problem $u(x; U, \p)$.
\end{remark}

\bibliography{stability}

\begin{thebibliography}{BCKMK13}

\bibitem[BCKMK13]{BCKMK13}
G.~Benedetti, L.~Campi, J.~Kallsen, and J.~Muhle-Karbe.
\newblock {On the existence of shadow prices}.
\newblock {\em Finance and Stochastics}, 17(4):801--818, 2013.

\bibitem[Ben12]{Ben12}
C.~Bender.
\newblock {Simple arbitrage}.
\newblock {\em Annals of Applied Probability}, 22(5):2067--2085, 2012.

\bibitem[BK13]{BK13}
E.~Bayraktar and R.~Kravitz.
\newblock {Stability of exponential utility maximization with respect to market
  perturbations}.
\newblock {\em Stochastic Processes and their Applications}, 123(5):1671--1690,
  2013.

\bibitem[BM03]{BM03}
B.~Bouchard and L.~Mazliak.
\newblock {A multidimensional bipolar theorem in $L^0(\mathbb
  R^d;\Omega;\mathcal F; \bf P)$}.
\newblock {\em Stochastic Processes and their Applications}, 107(2):213--231,
  2003.

\bibitem[Bou02]{Bou02}
B.~Bouchard.
\newblock {Utility maximization on the real line under proportional transaction
  costs}.
\newblock {\em Finance and Stochastics}, 6(4):495--516, 2002.

\bibitem[CK96]{CK96}
J.~Cvitani{\'c} and I.~Karatzas.
\newblock {Hedging and portfolio optimization under transaction costs: a
  martingale approach}.
\newblock {\em Mathematical Finance}, 6(2):133--165, 1996.

\bibitem[CMKS14]{CMKS14}
C.~Czichowsky, J.~Muhle-Karbe, and W.~Schachermayer.
\newblock {Transaction costs, shadow prices, and duality in discrete time}.
\newblock {\em SIAM Journal on Financial Mathematics}, 5(1):258--277, 2014.

\bibitem[CO11]{CO11}
L.~Campi and M.P. Owen.
\newblock {Multivariate utility maximization with proportional transaction
  costs}.
\newblock {\em Finance and Stochastics}, 15(3):461--499, 2011.

\bibitem[CPSY17]{CPSY16}
C.~Czichowsky, R.~Peyre, W.~Schachermayer, and J.~Yang.
\newblock {Shadow prices, fractional Brownian motion, and portfolio
  optimisation under transaction costs}.
\newblock {\em Preprint}, 2017.

\bibitem[CR07]{CR07}
L.~Carassus and M.~R{\'a}sonyi.
\newblock {Optimal strategies and utility-based prices converge when agents'
  preferences do}.
\newblock {\em Mathematics of Operations Research}, 32(1):102--117, 2007.

\bibitem[CS06]{CS06}
L.~Campi and W.~Schachermayer.
\newblock {A super-replication theorem in Kabanov's model of transaction
  costs}.
\newblock {\em Finance and Stochastics}, 10(4):579--596, 2006.

\bibitem[CS16a]{CS16duality}
C.~Czichowsky and W.~Schachermayer.
\newblock {Duality theory for portfolio optimisation under transaction costs}.
\newblock {\em Annals of Applied Probability}, 26(3):1888--1941, 2016.

\bibitem[CS16b]{CS16strong}
C.~Czichowsky and W.~Schachermayer.
\newblock {Strong supermartingales and limits of non-negative martingales}.
\newblock {\em Annals of Probability}, 44(1):171--205, 2016.

\bibitem[CSY17]{CSY17}
C.~Czichowsky, W.~Schachermayer, and J.~Yang.
\newblock {Shadow prices for continuous processes}.
\newblock {\em to appear in Mathematical Finance}, 2017.

\bibitem[CW01]{CW01}
J.~Cvitani{\'c} and H.~Wang.
\newblock {On optimal terminal wealth under transaction costs}.
\newblock {\em Journal of Mathematical Economics}, 35(2):223--231, 2001.

\bibitem[DM82]{DM82}
C.~Dellacherie and P.A. Meyer.
\newblock {\em {Probabilities and {P}otential. {B}. {T}heory of
  {M}artingales}}.
\newblock North-Holland Publishing Co., Amsterdam, 1982.

\bibitem[DPT01]{DPT01}
G.~Deelstra, H.~Pham, and N.~Touzi.
\newblock {Dual formulation of the utility maximization problem under
  transaction costs.}
\newblock {\em Annals of Applied Probability}, 11(4):1353--1383, 2001.

\bibitem[DS94]{DS94}
F.~Delbaen and W.~Schachermayer.
\newblock {A general version of the fundamental theorem of asset pricing}.
\newblock {\em Mathematische Annalen}, 300(3):463--520, 1994.

\bibitem[Fre13]{Fre13}
C.~Frei.
\newblock {Convergence results for the indifference value based on the
  stability of BSDEs}.
\newblock {\em Stochastics}, 85(3):464--488, 2013.

\bibitem[GLY16]{GLY16}
L.~Gu, Y.~Lin, and J.~Yang.
\newblock {On the dual problem of utility maximization in incomplete markets}.
\newblock {\em Stochastic Processes and their Applications}, 126(4):1019--1035,
  2016.

\bibitem[Gua02]{Gua02}
P.~Guasoni.
\newblock {Optimal investment with transaction costs and without
  semi-martingales.}
\newblock {\em Annals of Applied Probability}, 12(4):1227--1246, 2002.

\bibitem[Gua06]{Gua06}
P.~Guasoni.
\newblock {No arbitrage under transaction costs, with fractional Brownian
  motion and beyond}.
\newblock {\em Mathematical Finance}, 16(3):569--582, 2006.

\bibitem[JK95]{JK95}
E.~Jouini and H.~Kallal.
\newblock {Martingales and arbitrage in securities markets with transaction
  costs}.
\newblock {\em Journal of Economic Theory}, 66(1):178--197, 1995.

\bibitem[JN04]{JN04}
E.~Jouini and C.~Napp.
\newblock {Convergence of utility functions and convergence of optimal
  strategies}.
\newblock {\em Finance and Stochastics}, 8(1):133--144, 2004.

\bibitem[KMK11]{KMK11}
J.~Kallsen and J.~Muhle-Karbe.
\newblock {Existence of shadow prices in finite probability spaces}.
\newblock {\em Mathematical Methods of Operations Research}, 73(2):251--262,
  2011.

\bibitem[KS99]{KS99}
D.~Kramkov and W.~Schachermayer.
\newblock {The asymptotic elasticity of utility functions and optimal
  investment in incomplete markets}.
\newblock {\em Annals of Applied Probability}, 9(3):904--950, 1999.

\bibitem[Kv11]{KZ11}
C.~Kardaras and G.~\v{Z}itkovi{\'c}.
\newblock {Stability of utility maximization problem with random endowment in
  incomplete markets}.
\newblock {\em Mathematical Finance}, 21(2):313--333, 2011.

\bibitem[Lar09]{Lar09}
K.~Larsen.
\newblock {Continuity of utility - maximization with respect to preferences}.
\newblock {\em Mathematical Finance}, 19(2):237--250, 2009.

\bibitem[Lv07]{LZ07}
K.~Larsen and G.~\v{Z}itkovi{\'c}.
\newblock {Stability of utility-maximization in incomplete markets}.
\newblock {\em Stochastic Processes and their Applications},
  117(11):1642--1662, 2007.

\bibitem[Mer72]{Mer72}
J.-F. Mertens.
\newblock {Th{\'e}orie des processus stochastiques g{\'e}n{\'e}raux
  applications aux surmartingales}.
\newblock {\em Zeitschrift f{\"u}r Wahrscheinlichkeitstheorie und Verwandte
  Gebiete}, 22(1):45--68, 1972.

\bibitem[MW13]{MW13}
M.~Mocha and N.~Westray.
\newblock {The stability of the constrained utility maximization problem: a
  BSDE approach}.
\newblock {\em SIAM Journal on Financial Mathematics}, 4(1):117--150, 2013.

\bibitem[Roc70]{Roc70}
R.~T. Rockafellar.
\newblock {\em {Convex analysis}}.
\newblock Princeton university press, 1970.

\bibitem[RW98]{RW98}
R.T. Rockafellar and R.J.-B. Wets.
\newblock {\em {Variational analysis}}.
\newblock Springer-Verlag, Berlin, 1998.

\bibitem[Sch17]{Sch17}
W.~Schachermayer.
\newblock {\em {The asymptotic theory of transaction costs}}.
\newblock {Lecture notes Universit{\"a}t Wien \& ITS-ETH Z{\"u}rich}. European
  Mathematical Society, Zurich, 2017.

\bibitem[SW77]{SW77}
G.~Salinetti and R.J.-B. Wets.
\newblock {On the relations between two types of convergence of convex
  functions}.
\newblock {\em Journal of Mathematical Analysis and Applications},
  60(1):211--226, 1977.

\bibitem[Wes16]{Wes16}
K.~Weston.
\newblock {Stability of utility maximization in nonequivalent markets}.
\newblock {\em Finance and Stochastics}, 20(2):511--541, 2016.

\bibitem[Xin17]{Xin17}
H.~Xing.
\newblock {Stability of the exponential utility maximization problem with
  respect to preferences}.
\newblock {\em Mathematical Finance}, 27(1):38--67, 2017.

\bibitem[Yu17]{Yu17}
X.~Yu.
\newblock {Optimal consumption under habit formation in markets with
  transaction costs and random endowments}.
\newblock {\em The Annals of Applied Probability}, 27(2):960--1002, 2017.

\end{thebibliography}
\bibliographystyle{alpha}

\end{document}